%% file: main.tex
\newif\ifANON
\ANONtrue
\ANONfalse

\documentclass[11pt,a4]{article}

\usepackage{paper}

\title{The Power of Recursive Embeddings for $\ell_p$ Metrics} 

\ifANON
    \author{Anonymous}
    \date{}
\else
    \author{Robert Krauthgamer%
      \thanks{Work partially supported by the Israel Science Foundation grant \#1336/23,
        by the Israeli Council for Higher Education (CHE) via the Weizmann Data Science Research Center,
        and by a research grant from the Estate of Harry Schutzman.
        Email: \texttt{robert.krauthgamer@weizmann.ac.il}
      } 
      \qquad 
    Nir Petruschka\thanks{
        Email: \texttt{nir.petruschka@weizmann.ac.il}
      } 
      \qquad
    Shay Sapir\thanks{
        Email: \texttt{shay.sapir@weizmann.ac.il}
      } 
    \\  Weizmann Institute of Science
    }
\fi

\begin{document}

\maketitle

\input{abstract}

\input{intro}

\input{preliminaries}

\input{Iterative_Proof}

\input{NNS}

\input{l_2-distortion}

\input{Future_Directions}

{\small
  \bibliographystyle{alphaurl}
  \bibliography{references}
} %

\end{document}

%% file: abstract.tex
\begin{abstract}
Metric embedding is a powerful tool used extensively in mathematics and computer science.
We devise a new method of using metric embeddings recursively,
which turns out to be particularly effective in $\ell_p$ spaces, $p>2$,
yielding state-of-the-art results for Lipschitz decomposition,
for Nearest Neighbor Search, and for embedding into $\ell_2$.
In a nutshell, our method composes metric embeddings
by viewing them as reductions between problems,
and thereby obtains a new reduction that is substantially more effective
than the known reduction that employs a single embedding.
We in fact apply this method recursively, oftentimes using double recursion,
which further amplifies the gap from a single embedding.
\end{abstract}

%% file: intro.tex
\section{Introduction}

Metric embeddings represent points in one metric space 
using another metric space, often one that is simpler or easier,
while preserving pairwise distances within some distortion bounds.
This mathematical tool is very powerful at transferring properties
between the two metric spaces,
and is thus used extensively in many areas of mathematics and computer science. 
Its huge impact over the past decades is easily demonstrated
by fundamental results, such as John's ellipsoid theorem~\cite{John48},
the Johnson-Lindenstrauss (JL) Lemma~\cite{JL84},
Bourgain's embedding~\cite{Bourgain85},
and probabilistic tree embedding~\cite{Bartal96}. %

We devise a new method of using metric embeddings \emph{recursively},
in a manner that is particularly effective for $\ell_p$ spaces, $p>2$.
Our method is based on the well-known approach of embedding $\ell_p$ into $\ell_2$
(via the so-called Mazur map),
but leverages a new form of recursion that goes through intermediate spaces,
to beat a direct embedding from $\ell_p$ into $\ell_2$.

Our method is inspired by the concept of reduction between (computational) problems,
which is fundamental in computer science and has been used extensively
to design algorithms and/or to prove conditional hardness.
Many known reductions use metric embeddings in a straightforward manner,
without harnessing the full power of reductions,
which allow further manipulation,
like employing multiple embeddings and taking the majority (or best) solution.%
\footnote{This is perhaps analogous to the difference between Cook reductions and Karp reductions.
  The former allows the use of a subroutine that solves the said problem,
while the latter applies only a single transformation on the input,
and is thereby restricted to a single subroutine call.
}
To see this gap between embeddings and reductions,
consider a composition of multiple embeddings,
which yields overall an embedding from the first metric space to the last one.
While going through intermediate metric spaces may simplify the exposition,
it can only restrict the overall embedding.
In contrast, composing metric embeddings by way of reductions,
can create new reductions that are substantially richer than any single direct embedding.
Our method actually composes reductions \emph{recursively},
which makes this gap even more pronounced. 
We emphasize that the application of this method is problem-specific,
unlike a metric embedding which is very general and thus applies to many problems at once.
On the flip side, tailoring our recursive method to a specific problem
opens the door to embeddings that are non-oblivious to the problem/data, 
which is reminiscent of data-dependent space partitioning
used in recent nearest neighbor search (NNS) algorithms~\cite{ANNRW18a, ANNRW18b,KNT21}. 
To the best of our knowledge, this recursive method is new,
i.e., related to but different from variants that have been used in prior work.

Our method yields several state-of-the-art results:
(i) Lipschitz decomposition for finite subsets of $\ell_p$ spaces, $p>2$; 
(ii) consequently, also Lipschitz decomposition for $\ell_\infty^d$;
and 
(iii) algorithms for NNS in $\ell_p$ spaces, $p>2$. 
After obtaining these results,
we noticed the online posting of parallel work~\cite{NR25},
and realized that our method can also
(iv) improve some of its results about embedding into $\ell_2$.

\subsection{Lipschitz Decomposition}

A standard approach in many metric embeddings and algorithms
is to partition a metric space into low-diameter (so-called) clusters,
and the following probabilistic variant is commonly used and highly studied
(sometimes called a separating decomposition).

\begin{definition}[Lipschitz decomposition \cite{Bartal96}]
Let $(\M,d_\M)$ be a metric space.
A distribution $\D$ over partitions of $\M$ is called
a \emph{$(\beta,\Delta)$-Lipschitz decomposition}
if
\begin{itemize} \compactify
\item
  for every partition $P\in \supp(\D)$, all clusters $C\in P$
  satisfy $\diam(C) \leq \Delta$;
  and
\item
  for every $x,y\in \M$,
  \[
    \Pr_{P\sim \D}[P(x)\neq P(y)]\leq \beta \tfrac{d_\M(x,y)}{\Delta},
  \]
  where $P(z)$ denotes the cluster of $P$ containing $z\in \M$
  and $\diam(C) \coloneqq \sup_{x,y\in C}d_\M(x,y)$.
\end{itemize}
\end{definition}
Our first use of recursive embedding yields the following theorem, 
whose proof appears in \Cref{sec:iterative_lipschitz_decomp}.
\begin{theorem}\label{thm:Optimal-Lipschitz-Decompositions-in-lp}
Let $p \geq 2$ and $d \geq 1$.
Then for every \( n \)-point metric \( \C \subset \ell_{p}^{d} \) and $\Delta>0$,
there exists an 
$(O(p^4\sqrt{\min\{\log {n}, d\}}), \Delta)$-Lipschitz decomposition.
\end{theorem}

Typically, $\Delta$ is not known in advance
or one needs multiple values of $\Delta$ (e.g., every power of $2$).
We naturally  seek the smallest possible $\beta$ in this setting,
and thus define the (optimal) \emph{decomposition parameter} of a metric space $(\M,\rho)$ as
\[
  \beta^*(\M) \coloneqq \inf_{\beta\ge1}
  \Big\{\beta :\
    \text{$\forall \Delta>0$, every finite $\M' \subseteq \M$
      admits a $(\beta, \Delta)$-Lipschitz decomposition} 
  \Big\} , 
\]
and further define 
$\beta^*_n(\M) \coloneqq
  \sup \big\{ \beta^*(\M') :\ \M'\subseteq \M, |\M'| \leq n \big\} $.
The following two corollaries of \Cref{thm:Optimal-Lipschitz-Decompositions-in-lp} 
bound these quantities and delineate the asymptotic dependence on $n$ and on $d$.

\begin{corollary}
\label{cor:Optimal-finite-l_p_Lipschitz-Decomposition}
For every $p\in[2,\infty)$ and $n \geq 1$, we have
$\beta^{*}_{n}(\ell_{p})=O(p^4\sqrt{\log{n}})$.
\end{corollary}

\begin{proof}
It follows directly from \Cref{thm:Optimal-Lipschitz-Decompositions-in-lp} and the result from \cite{Ball90},
that every finite set $X \subset \ell_p$ embeds isometrically into $\ell_p^d$ for some $d$.
\end{proof}

This result significantly improves the previous bound
$\beta^*_n(\ell_p) = O(\log^{1-1/p} n)$ from~\cite{KP25}, 
and fully resolves \cite[Question 1]{Naor17} (see also \cite[Question 83]{Naor24}),
which asked for an $O_p(\sqrt{\log n})$ bound.
(Throughout, the notation $O_\alpha(\cdot)$ hides a factor that depends only on $\alpha$.)
In parallel to our work, a slightly weaker bound 
$\beta^*_n(\ell_p) \leq O(2^p\sqrt{\log n})$
was obtained in~\cite{NR25}.
Both our improvement and that of~\cite{NR25}
rely on the technique developed in~\cite{KP25},
and essentially apply it iteratively/recursively instead of once,
and ours actually applies double recursion.

\begin{corollary}
\label{cor:Optimal-l_p_d_Lipschitz-Decomposition}
For every $p\in[2,\infty]$ and $d \geq 1$, we have
    $\beta^*(\ell_{p}^{d})=O( (\min\{p,\log{d}\})^4 \cdot \sqrt{d})$.
\end{corollary}
\begin{proof}
For $p \leq \log d$,
it follows from \Cref{thm:Optimal-Lipschitz-Decompositions-in-lp}. 
For larger $p$, use Hölder's inequality to reduce the problem
from $\ell_{p}^{d}$ to $\ell_{\log d}^{d}$ with $O(1)$ distortion.%
\footnote{A metric space $(\M, d_\M)$ embeds into a metric space $(\mathcal{N}, d_\mathcal{N})$ with distortion $D \geq 1$ iff there exists \( s > 0 \) and a function $f: \M \to \mathcal{N}$ such that for all \( x, y \in \M\), $\frac{s}{D} \cdot d_\M(x, y) \leq d_\mathcal{N}(f(x), f(y)) \leq s \cdot d_\M(x, y)$.}
\end{proof}

\Cref{cor:Optimal-l_p_d_Lipschitz-Decomposition} is slightly weaker than
Naor's main result in~\cite{Naor17} (see also~\cite{Naor24}). Naor showed that $\beta^*(\ell_p^d) = O(\sqrt{\min\{p, \log d\} \cdot d})$
for all $p \in [2, \infty]$,
which nearly matches the $\Omega(\sqrt{d})$ lower bound that follows from~\cite{CCGGP98}. Our proof is fundamentally different from, and arguably simpler than, Naor’s proof,
which relies on a deep understanding of the geometry of $\ell_p^d$ spaces.
One may hope that our proof could resolve the exact asymptotics of $\beta^*(\ell_\infty^d)$,
perhaps by simply optimizing the constants in our recursion that yield
the $p^4$ factor in \Cref{thm:Optimal-Lipschitz-Decompositions-in-lp}.
Unfortunately, this approach has a serious barrier.
For $\ell_{\log n}$, we have $\beta^{*}_{n}(\ell_{\log n}) = \Omega(\log n)$, since every $n$-point metric embeds into $\ell_{\log n}$ with $O(1)$ distortion by~\cite{Matousek97}, and there is an $\Omega(\log n)$ lower bound for Lipschitz decomposition of general $n$-point metrics~\cite{Bartal96}.
Improving the $p^4$ factor in our analysis to $o(\sqrt{p})$ would imply that $\beta^{*}_{n}(\ell_{\log n}) = o(\log n)$, contradicting the known lower bound.

\begin{remark}
Naor~\cite{Naor17} shows that his upper bound on $\beta^*(\ell_\infty^{d})$
has an important application to the Lipschitz extension problem. 
More precisely, he proves an infinitary variant of his upper bound,
and that it implies a similar bound on $e(\ell_{\infty}^{d})$,
which is the Lipschitz extension modulus of $\ell_\infty^d$.
He thus concludes that $e(\ell_{\infty}^{d}) \leq O(\sqrt{d\log d})$,
which almost matches (up to lower order factors), the lower bound
$e(\ell_{\infty}^{d}) \geq \Omega(\sqrt{d})$ that follows from~\cite{BB05, BB06}.
We have not attempted to extend \cref{cor:Optimal-l_p_d_Lipschitz-Decomposition}
to the infinitary variant,
as Naor notes that it is required only for extension theorems
into certain exotic Banach spaces \cite[Appendix A, Remark 4]{Naor17}. 
\end{remark}

\begin{remark}
The result of \Cref{thm:Optimal-Lipschitz-Decompositions-in-lp} extends to a related notion of decomposition, 
that was introduced in~\cite{FN22} and immediately implies geometric spanners.
This yields spanners for $\ell_p$ spaces, $p > 2$,
whose stretch-size tradeoff is comparable to that known for $\ell_2$.
Previously, weaker bounds for such decompositions, 
and consequently also weaker spanners for $\ell_p$, 
were proved in~\cite{KP25}. 
The details, which are similar to \Cref{thm:Optimal-Lipschitz-Decompositions-in-lp}, are omitted.
\end{remark}

\subsection{Nearest Neighbor Search}

The Nearest Neighbor Search (NNS) problem is to design a data structure
that preprocesses an $n$-point dataset $V$ residing in a metric $\M$,
so that given a query point $q\in \M$,
the data structure reports a point in $V$ that is closest to $q$
(and approximately closest to $q$ in approximate NNS). 
The main measures for efficiency
are the data structure's space complexity and the time it takes to answer a query; a secondary measure is the preprocessing time, which is often proportional to the space.
The problem has a wide range of applications in machine learning, computer vision and other fields,
and has thus been studied extensively,
including from theoretical perspective, see e.g.\ the survey \cite{AndoniIndyk}.
It is well known that approximate NNS reduces to solving $\polylog(n)$ instances
of the approximate \emph{near} neighbor problem~\cite{IM98},
hence we consider the latter.

\begin{definition}[Approximate Near Neighbor]
The Approximate Near Neighbor problem for a metric space \((\M, d_\M)\)
and parameters $c \geq 1$, $r>0$, abbreviated \emph{\((c,r)\)-ANN},
is the following. 
Design a data structure 
that preprocesses an \(n\)-point subset \(V \subseteq \M\), 
so that given a query \(q \in \M\) with $d_\M(q,V) \leq r$,%
\footnote{If $d_\M(q,V) > r$, it may report anything,
  where as usual, $d_\M(q,V)\coloneqq \min_{x^* \in V} d_\M(x^*, q)$. 
}
it reports \(x \in V\) such that
\[
  d_\M(q, x) \leq cr.
\]
In a randomized data structure,
the reported $x$ satisfies this with probability at least $2/3$. 
\end{definition}

We prove the following theorem, whose 
proof appears in \Cref{sec:ANN} and is similar in spirit
to that of \Cref{thm:Optimal-Lipschitz-Decompositions-in-lp}.
It applies our method of recursive embedding,
using Mazur maps for $n$-point subsets of $\ell_p^d$.

\begin{theorem}
\label{thm:BNN-for-l_p}
Let $p>2$, $d\ge 1$ and $0<\eps<1$.
Then for $c=O(p^{1+\ln 4+\eps})$ and every $r>0$,
there is a randomized data structure for $(c,r)$-ANN in $\ell_p^d$, 
that has query time $\poly(\eps^{-1}d\log n)$, 
and has space and preprocessing time
$\poly(d n^{\eps^{-1}\log p} )$.
\end{theorem}

\begin{remark*}
Picking $\eps=\tfrac{1}{\log p}$ is sufficient to get approximation $O(p^{1+\ln 4}) \leq O(p^{2.387})$. 
\end{remark*}

Most prior work on ANN in $\ell_p$ spaces
studies the case $1\leq p \leq 2$, where $(O(1),r)$-ANN can be solved using
query time $\poly(d\log n)$ and space $\poly(n)$~\cite{KOR00,IM98,HIM12}.
For $p>2$, such a bound is not known,
and we list in \Cref{table:ANN_related} all the known results (ours and previous ones),
which are often incomparable.
The results of \cite{Andoni09_thesis,AIK09} and of~\cite{ANRW21}
are based on Indyk's~\cite{Indyk01_ell_infty} result for $\ell_\infty$, 
and are most suitable for large values of $p$;
note though that the \emph{preprocessing time} of~\cite{ANRW21} is exponential in $d$.
The other results are more suited for small values of $p>2$, 
and they all have different downsides:
one result~\cite{BG19} has a large approximation $2^{O(p)}$; 
another one~\cite{ANNRW18a, ANNRW18b, KNT21} has a large query time $n^\eps \cdot \poly(d\log n)$,
which can be mitigated by picking $\eps = \tfrac{1}{\log n}$,
at the cost of increasing the approximation to $O(p\log n)$; 
ours (\Cref{thm:BNN-for-l_p}) has a large space $n^{O(\log p)}$; 
and lastly, \cite{BBMWWZ24} and \cite{AIK09, Andoni09_thesis}
can achieve $O(1)$-approximation
but this requires an even larger space $d\cdot n^{ 2^{O(p)}\log(1/\eps)}$ and $n^{O(\log d)}$, respectively.
The bottom line is that the regime of $p>2$ is notoriously difficult.
It remains open to bridge the gap between small $p$ and large $p$,
and specifically to obtain $O(p)$-approximation using 
$\poly(d\log n)$ query time and $\poly(n)$ space.

Our result for ANN provides yet another illustration for the power of recursive embedding.
Bartal and Gottlieb~\cite{BG19} mentioned that Assaf Naor noted,
in personal communication regarding improving their $2^{O(p)}$-approximation,
that all uniform embeddings of $\ell_p$ to $\ell_2$ (like Mazur maps) have distortion exponential in $p$~\cite[Lemma 5.2]{Naor14}.
Our use of recursive embeddings breaks this barrier, 
and essentially provides a black-box reduction from $\ell_p$ to $\ell_2$,
that still uses Mazur maps but achieves $\poly(p)$-approximation.
We note that the improved approximation of~\cite{ANNRW18a, ANNRW18b, KNT21} 
uses embedding into $\ell_2$ with small average distortion, 
however this approach is not known to provide a black-box reduction for ANN, 
and its specialized solution increases the query time.

\begin{table}[t]
\begin{center}
\begin{tabular}{llll}
\midrule
Approximation  & Query time & Space & Reference  \\
\midrule
\midrule
$O(\eps^{-1}\log\log d)$ & $n^{\eps^p}$   & $n^{1+\eps}$ & \cite{AIK09, Andoni09_thesis}  \\
$O_\eps(\log p \cdot(\log d)^{2/p})$ & $n^{\eps}$   & $n^{1+\eps}$ & \cite{ANRW21}  \\
\midrule
$2^{O(p)}$     & $(d\log n)^{O(1)}$ & $n^{O(1)}$ & \cite{BG19} \\
$p^{O(1)}$     & $(d\log n)^{O(1)}$ & $n^{O(\log p)}$ & Thm~\ref{thm:BNN-for-l_p}  \\
$O(p/\eps)$    & $n^\eps$           & $n^{1+\eps}$                                 & \cite{ANNRW18a, ANNRW18b, KNT21} \\
$c$     & $n^\eps$ & $n^{ (cp \cdot 2^{p/c})^{O(1)}\log(1/\eps)}$ & \cite{BBMWWZ24}  \\
\midrule
\end{tabular}
\end{center}
\caption{Known data structures for ANN in $\ell_p$, $p>2$.
  For brevity, we omit here $\poly(d\log n)$ factors when the complexity is polynomial in $n$.
  The top-listed two results are particularly suited for large values of $p$, and the others are suited for small values of $p$.
}
\label{table:ANN_related}
\end{table}

\subsection{Low-Distortion Embeddings}
After we obtained our aforementioned results for Lipschitz decomposition and NNS, we noticed the online posting of \cite{NR25} on the distortion required for embedding $\ell_p$ space ($p>2$) into Euclidean space, and used our technique to extend their result.
The study of the distortion required for embedding metrics into Euclidean space has a decades-long history for general metrics~\cite{John48, Bourgain85, LLR95} and for $\ell_p$ space~\cite{Lee05, CGR05, ALN08, CNR24, BG14_v2, NR25}.
For an infinite metric space \( (\M, d_\M) \), define $c_2^n(\M) \coloneqq \sup_{\mathcal{C} \subseteq \M, \, |\mathcal{C}| \leq n} c_2(\mathcal{C})$, 
where $c_2(\mathcal{C})$ denotes the minimal distortion needed to embed $\mathcal{C}$ into $\ell_2$. 
We prove the following in \Cref{sec:Embedding-finite-l_p-metrics-into-l_2}.

\begin{theorem}\leavevmode\hspace*{-1.1in}\phantomsection\hspace{1.1in}\label{thm:improved-c_2-distortion}
If \( 3 < p < 3\sqrt{e} \), then for every fixed $0< \eps \leq 1$,

\[
c_{2}^{n}(\ell_p) \leq O(\log^{\frac{1}{2}+\ln{\frac{p}{3}} + \varepsilon}{n}).
\]
\end{theorem}
Previously, for $p>2$, non-trivial distortion was only known in the range 
$2<p<4$~\cite{BG14_v2,NR25}, where non-trivial means distortion asymptotically smaller than \( O(\log{n}) \), which holds for every \( n \)-point metric space~\cite{Bourgain85}. 
Bartal and Gottlieb~\cite{BG14_v2} established that \( c_2^n(\ell_p) = O(\log^{p/4}{n}) \) for every \( p \in (2,4) \), and
Naor and Ren~\cite{NR25} proved a better bound \( c_2^n(\ell_p) = O(\sqrt{\log{n}} \cdot\log \log {n}) \) for \( p \in (2,3] \) and \( c_2^n(\ell_p) = O(\log^{p/2-1}{n} \cdot\log \log {n}) \) for \( p \in (3,4) \). 
\Cref{thm:improved-c_2-distortion} improves these bounds further 
in the range \( 3 < p < 3\sqrt{e} \).  
Since it may not be immediate that \Cref{thm:improved-c_2-distortion} indeed improves the bounds on \( c_{2}^{n}(\ell_p) \) for all \( 3<p<3\sqrt{e} \), 
we plot the corresponding exponents of the \( \log n \) factor in \Cref{fig:exponent_of_log_NR25}.

\begin{figure}
\begin{center}
    \includegraphics[width=0.5\linewidth]{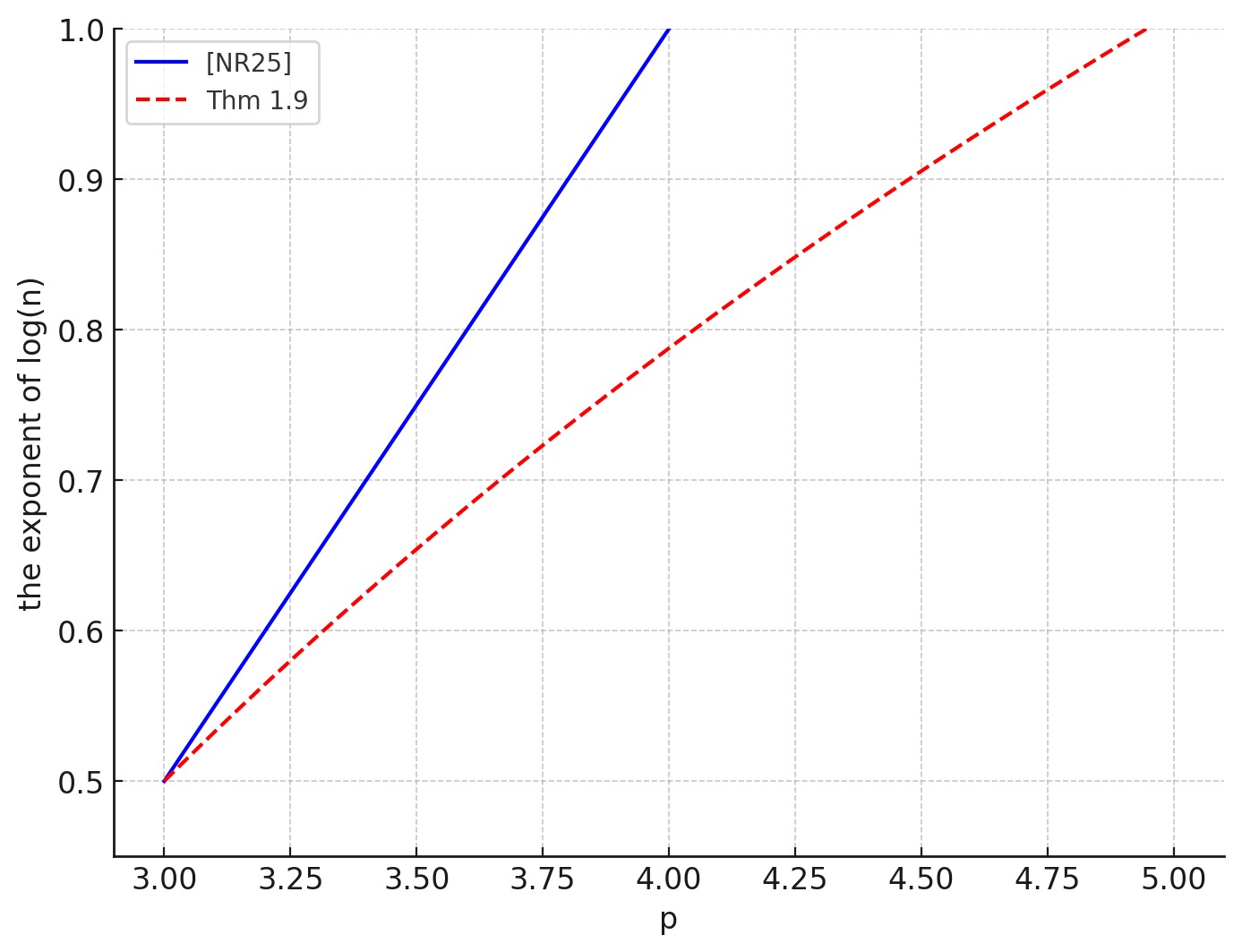}
    \caption{The distortion of embedding from $\ell_p$, $p>3$ into $\ell_2$
    shown by depicting the exponent of $\log n$ 
    in \cite[Theorem 1]{NR25} (blue) 
    compared with our bound in \Cref{thm:improved-c_2-distortion} (red).
    }\label{fig:exponent_of_log_NR25}
\end{center} 
\end{figure}

\begin{remark}
\label{rem:embedding-lower-bound}
Every finite metric embeds isometrically in $\ell_\infty$, 
and thus  $c_2^n(\ell_\infty)=\Theta(\log{n})$ by \cite{Bourgain85} and \cite{LLR95}. 
For $\ell_p$, $p \in (2, \infty)$, a lower bound of  
\[
c_2^n(\ell_p) \geq \Omega(\log^{1/2 - 1/p} n)
\]  
follows from \cite[Theorem 1.3]{LN13}.
\end{remark}

%% file: preliminaries.tex
\section{Preliminaries}

The main tool we use for recursive embeddings between \( \ell_p \) spaces is a classical embedding, commonly known as the Mazur map. For every \( p, q \in [1, \infty) \), the Mazur map  $M_{p,q}: \ell^m_p \to \ell^m_q$ is computed by raising the absolute value of each coordinate to the power \( p/q \) while preserving the original signs. The following key property of this map is central to all our results.

\begin{theorem} [\cite{benyamini1998geometric,BG19}]\label{thm:mazur_map}
Let $1 \leq q < p < \infty$ and $C_0>0$,
and let $M$ be the Mazur map $M_{p,q}$ scaled down by factor $\frac{p}{q}{C_0}^{p/q-1}$.
Then for all $x,y\in\ell_{p}$ such that $||x||_p,||y||_p\leq C_0$,
\[
  \tfrac{q}{p} (2 C_0 )^{1 - p/q} ||x - y||_{p}^{p/q}
  \leq ||M(x) - M(y)||_{q}
  \leq ||x - y||_{p} .
\]
\end{theorem}

%% file: Iterative_Proof.tex
\section{Lipschitz Decomposition of $\ell_p$ Metrics}
\label{sec:iterative_lipschitz_decomp}
In this section, we prove \Cref{thm:Optimal-Lipschitz-Decompositions-in-lp}. We first outline the proof.
Our approach uses a double recursion, where each recursion is an instance of recursive embedding.
The first recursion takes a Lipschitz decomposition of a finite subset $\M \subset \ell_{p}^{d}$ with decomposition parameter $\beta$ and produces a Lipschitz decomposition with (ideally smaller) decomposition parameter $\beta_{new}$.
Each iteration in this recursion is as follows.
We first use the given decomposition
to decompose $\M$ into bounded-diameter subsets, embed each subset into $\ell_{q}$ for $q<p$ using Mazur maps, employ Lipschitz decomposition for $\ell_{q}$, and pull back
the solution (clusters) we found.
It is natural to choose here $q=2$,
because the known Lipschitz decompositions for $\ell_2$ are tight.
However, this choice leads to a decomposition parameter with an $\exp(p)$ factor,
and we overcome this by picking $q=p/2$.
We only then apply a second recursion,
which goes from $\ell_p$ to $\ell_2$ gradually, via intermediate values $2<q<p$.

\begin{lemma}\label{lem:iterative_decomposition}
Let $2\leq q<p<\infty$ and let $\M\subset \ell_p$ be an $n$-point metric. Suppose that for every $\Delta'>0$, there exists a $(\beta,\Delta')$-Lipschitz decomposition of $\M$.
Then, for every $\Delta>0$, there exists a $(\beta_{new},\Delta)$-Lipschitz decomposition of $\M$, with
    \[
    \beta_{new} = 4(\tfrac{p}{2q})^{q/p} \ [\beta^{*}_{n}(\ell_q)]^{q/p} \ \beta^{1-q/p}.
    \]
\end{lemma}

\Cref{lem:iterative_decomposition} provides the recursion step for the first recursion from the outline above,
and we use it with $q=p/2$.
For the natural choice of $q=2$, the expression in \Cref{lem:iterative_decomposition} equals $\beta_{new} = 4 (p/4)^{2/p} \ [\beta^{*}_{n}(\ell_2)]^{2/p} \ \beta^{1-2/p}$, 
hence iterative applications converge to the fixpoint $\beta = \tfrac{p}{4} 2^p \cdot \beta^{*}_{n}(\ell_2)$,
which is easily found by setting $\beta=\beta_{new}$.
In contrast, for $q=p/2$, the expression simplifies to 
$\beta_{new} = 4\sqrt{\beta^{*}_{n}(\ell_{p/2})\cdot \beta}$, 
the fixpoint is now $\beta = 16 \beta^{*}_{n}(\ell_{p/2})$,
and recursion on $p$ introduces only a $\poly(p)$ factor.

\begin{proof}
    Let \( \Delta > 0, p \in (2, \infty) \), and let \(\M \subset \ell_{p} \) be an \( n \)-point metric space.
    Set $a\coloneqq\tfrac{1}{2} \big(\tfrac{2q\beta}{p\beta^{*}_{n}(\ell_q)}\big)^{q/p}$ and $b \coloneqq \tfrac{\beta^{*}_{n}(\ell_{q})a}{\beta}$, chosen to satisfy 
    \begin{equation}\label{eq:parameters_iteration}
        \tfrac{\beta}{a}=\tfrac{\beta^{*}_{n}(\ell_{q})}{b} \qquad \text{and} \qquad \tfrac{p}{q} (2a)^{p/q-1} b  = 1.
    \end{equation}
    Construct a partition of $\M$ in the following steps:
\begin{enumerate} \compactify
\item
  Draw a partition \( \Pinit  = \{ K_{1}, \ldots, K_{t} \} \) 
  from a \( (\beta, a \Delta) \)-Lipschitz decomposition of $\M$.
\item
  Embed each cluster $K_i \subset \ell_p$ into \( \ell_{q} \)
  using the embedding \( f^{K_{i}} \) provided by \Cref{thm:mazur_map} for $C_0 \coloneqq a \Delta$.
\item
   For each embedded cluster \( f^{K_{i}}(K_{i}) \), 
   draw a partition \( P_{i} = \{K_{i}^{1}, \ldots,K_{i}^{k_{i}} \} \) 
   from a \( (\beta^{*}_{n}(\ell_{q}), b \Delta) \)-Lipschitz decomposition of \( f^{K_{i}}(K_{i}) \).
\item 
  Obtain a final partition $\Pout$ by taking the preimage of every cluster of every $P_i$. 
\end{enumerate}

It is easy to see that $\Pout$ is indeed a partition of $\M$,
consisting of $\sum_{i=1}^t k_i$ clusters.
Next, consider $x,y\in \M$ and let us bound $\Pr[\Pout(x) \neq \Pout(y)]$. 
Observe that a pair of points can be separated only in steps 1 or 3. 
Therefore, 
\begin{align*}
  \Pr &\Big[\Pout(x) \neq \Pout(y)\Big]
  \\
  & \leq \Pr\Big[\Pinit (x) \neq \Pinit (y)\Big]
  +  \Pr\Big[P_i(f^{K_i}(x)) \neq P_i(f^{K_i}(y)) \mid \Pinit (x) = \Pinit (y) = K_{i}\Big]
  \\
  & \leq \beta \frac{\|x - y\|_{p}}{a \Delta} + \beta^{*}_{n}(\ell_{q}) \frac{\|f^{K_i}(x) - f^{K_i}(y)\|_{q}}{b \Delta }
  \\
  & \leq  \Big(\tfrac{\beta }{a} + \tfrac{\beta^{*}_{n}(\ell_{q})}{b}\Big)\frac{\|x - y\|_{p}}{\Delta}  ,
\end{align*}
where the last inequality is because by \Cref{thm:mazur_map},
each $f^{K_i}$ is a non-expanding map from $K_{i}\subset \ell_p$ to $\ell_q$.
Using~\eqref{eq:parameters_iteration}, we obtain
$\beta_{new}
  = 2\tfrac{\beta}{a}
  = 4(\tfrac{p}{2q})^{q/p}[\beta^{*}_{n}(\ell_q)]^{q/p}\beta^{1-q/p}$.

It remains to show that the final clusters all have diameter at most \( \Delta \).
Let \( x, y \in\M \) be in the same final cluster, i.e., \( \Pout(x) = \Pout(y) \).
Then $\Pinit (x)=\Pinit (y)=K_i$ and $P_{i}(f^{K_i}(x))=P_{i}(f^{K_i}(y))$.
Combining the distortion guarantees of \( f^{K_i} \) from \Cref{thm:mazur_map}
with the diameter bound of $P_i$, 
we get
\[
  \frac{q}{p} \Big( 2 a\Delta \Big)^{1-p/q} \|x - y\|_{p}^{p/q} 
  \leq \|f^{K_i}(x) - f^{K_i}(y)\|_{q} 
  \leq b\Delta.
\]
Rearranging this and using~\eqref{eq:parameters_iteration}, we obtain
\( \|x - y\|_p^{p/q} \leq \tfrac{p}{q} (2a)^{p/q-1} b \Delta^{p/q} = \Delta^{p/q} \),
which completes the proof.
\end{proof}

We are now ready to prove the main theorem.

\begin{proof}[Proof of \Cref{thm:Optimal-Lipschitz-Decompositions-in-lp}]
  Let \( p \in (2, \infty) \), and let \(\M \subset \ell_{p} \) be an \( n \)-point metric space.
  For ease of presentation, we assume for now that $p$ is a power of $2$,
  and resolve this assumption at the end.
    Denote $\beta_0(\M) = O(\min\{d, \log n\})$, given by \cite{Bartal96} and \cite{CCGGP98}.
    We now iteratively apply \Cref{lem:iterative_decomposition} with $q=p/2$, and obtain after $k$ iterations,
    \begin{align}
        \beta_k(\M) &= 4 \sqrt{\beta^{*}_{n}(\ell_{p/2})\cdot \beta_{k-1}(\M)} \nonumber \\
        &= 4 \sqrt{\beta^{*}_{n}(\ell_{p/2})\cdot 4 \sqrt{\beta^{*}_{n}(\ell_{p/2})\cdot \beta_{k-2}(\M)}} \nonumber \\
        &= \cdots \nonumber \\
        &= {4}^{(1 + 1/2 + \ldots + 1/2^{k-1})} \ [\beta^{*}_{n}(\ell_{p/2})]^{(1/2+1/4+\ldots+1/2^k)} \ \beta_0(\M)^{1/2^k} \nonumber \\
        &\leq 16 \beta^{*}_{n}(\ell_{p/2})\cdot \beta_0(\M)^{1/2^k}. \label{eq:recursion_p_to_p/2}
    \end{align}
Picking
$k := \ceil{\log(\log p\log \beta_0(\M))} = O( \log(\log p\log \min\{d,\log n\}))$
yields $\beta_0(\M)^{1/2^k} \leq 2^{1/\log p}$, 
and we obtain $\beta^{*}(\M)\leq \beta_k(\M) \leq 2^{4+1/\log p} \cdot \beta^{*}_{n}(\ell_{p/2})$.
Now recursion on $p$ implies
\[
  \beta^*(\M)\leq 2 p^4\cdot \beta^{*}_{n}(\ell_{2}).
\]
Finally, by \cite{CCGGP98} and the JL Lemma~\cite{JL84} we know that
$\beta^{*}_{n}(\ell_{2}^{d}) \leq O(\min\{\sqrt{d},\sqrt{\log n}\})$,
which concludes the proof when $p$ is a power of $2$.

Resolving the case when $p$ is not a power of $2$ is straightforward. 
Let $q$ be the largest power of $2$ that is smaller than $p$, hence $1/2<q/p<1$. 
It suffices to show that $\beta^{*}_{n}(\ell_{p})= O(\beta^{*}_{n}(\ell_{q}))$,
as then we can apply the previous argument since $q$ is a power of $2$.
Now apply \Cref{lem:iterative_decomposition} for $k$ iterations,
analogously to~\eqref{eq:recursion_p_to_p/2}.
We may assume that $\beta^{*}_{n}(\ell_q)\leq \beta_{i}(\M)$ for all $i\leq k$,
as otherwise we can simply abort after the $i$-th iteration,
hence
$\beta_{k}(\M) = 4(\tfrac{p}{2q})^{q/p} \ [\beta^{*}_{n}(\ell_q)]^{q/p} \ \beta_{k-1}(\M)^{1-q/p}\leq 4 \sqrt{\beta^{*}_{n}(\ell_q) \ \beta_{k-1}(\M)}$.
Now similarly to~\eqref{eq:recursion_p_to_p/2} 
we get $\beta^{*}_{n}(\ell_{p})= O(\beta^{*}_{n}(\ell_{q}))$,
and the theorem follows. 
\end{proof}

\begin{remark}
\label{rem:refined-iterative-solution}
    We suspect that the factor $16$ in the recursion~\eqref{eq:recursion_p_to_p/2} is an artifact of the analysis. 
    First, by balancing the separation probabilities over all $k$ iterations, one can perhaps eliminate the factor $2$ increase in the probabilities, and thus improve the factor in the recursion to roughly $4$. Second, the Mazur maps require sets of bounded \emph{radius}, while the construction guarantees sets of bounded \emph{diameter}. 
    Our proof uses the trivial bound $\radius\leq \diam$, which holds for every metric space,
    and subsets of $\ell_p$ may admit a tighter bound. 
    Denote by $J_p\in [\tfrac{1}{2},1]$ the minimum number such that $\radius(\M)\leq J_p \diam(\M)$ for all $\M\subset \ell_p$.
    It is known that $J_\infty=1/2$ and by Jung's Theorem, $J_2= \tfrac{1}{\sqrt{2}}$.
    Then, the factor above improves to roughly $(2J_p)^{2}$. Keeping in mind the discussion following \Cref{cor:Optimal-l_p_d_Lipschitz-Decomposition}, and aiming for a clear presentation of the main ideas in the solution, we have omitted the above optimizations.
\end{remark}

%% file: NNS.tex
\section{Nearest Neighbor Search}\label{sec:ANN}
In this section, we design a data structure for approximate NNS in \( \ell_p^d \) for \( p > 2 \), proving \Cref{thm:BNN-for-l_p}.
Previously, Bartal and Gottlieb~\cite{BG19} devised a data structure
that is based on embedding $\ell_p$ into $\ell_2$, for which good data structures are known (e.g., LSH),
and they furthermore employ recursion to improve the approximation factor, 
from a large trivial factor down to $\exp(p)$. 
We observe that their embedding and recursion approach 
is actually analogous to \Cref{sec:iterative_lipschitz_decomp}, but using only the special case $q=2$. 
We thus use our double recursion approach that goes through intermediate \( \ell_q \) spaces, 
and obtain an improved approximation factor $\poly(p)$.
In the rest of this section, we reserve the letter $q$ for the query point
(which is standard in the NNS literature) and denote the intermediate spaces by $\ell_t$.

\begin{proof}[Proof of \Cref{thm:BNN-for-l_p}]
First, we show an analogous claim to \Cref{lem:iterative_decomposition} but for the $(c,r)$-ANN problem. 
We take two NNS data structures, one for $\ell_p^d$ with approximation $c_p$ 
and one for $\ell_t^d$ (where $t<p$) with approximation $c_t$, 
and construct a new data structure for $\ell_p^d$ with approximation $c_{new}$ (ideally smaller than $c_p$).

Given an $n$-point dataset $V\subset\ell_p^d$, construct a $(c_p,r)$-ANN $A_{base}$ for $V$; and additionally, for every point $x\in V$, apply a Mazur map $M^{x}$ scaled down by $\frac{p}{t} \cdot (2 r c_{p})^{p/t-1}$ from $\ell_{p}^{d}$ to $\ell_{t}^{d}$ on $B_p(0,2rc_p)\cap (V-x)$, where $B_p(x,r) \coloneqq \{y : \|x-y\|_p\leq r\}$, and construct a $(c_t, r)$-ANN data structure $A_x$ for the image points.
Amplify their success probabilities to $5/6$ by standard amplification.
Given a query $q$, with the guarantee that there exists $x^* \in V$ with $\|x^*- q\|_p \leq r$, query $A_{base}$ with $q$ and obtain a point $x \in V$. 
Then query $A_x$ with $M^{x}(q-x)$, obtain a point $M^{x}(z-x) \in M^x(V-x)$ and output $z$ accordingly. 

\begin{figure}[!htbp]
\centering
  \includegraphics[width=15cm]{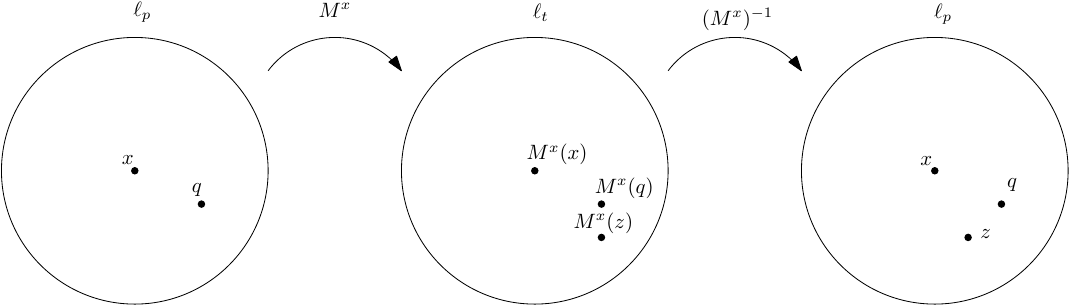}
   \caption{An illustration of \Cref{claim:ANN_Mazur_Map}. For the purpose of this illustration, the $\ell_p$ and $\ell_t$ balls are depicted using a Euclidean circle, and $x$ is assumed to lie at the origin of $\ell_p$. Given a query point $q$, an approximated solution $x$ is found in $\ell_p$ using $A_{base}$. The Mazur map $M^x$ is then applied, after which a solution $M^x(z)$ is found in $\ell_t$ using $A_x$. Finally, the inverse map is applied to obtain an improved solution $z$ in $\ell_p$.}
\end{figure}

\begin{claim}\label{claim:ANN_Mazur_Map}
    With probability $2/3$, we have $\|z-q\|_p\leq c_{new} r$, where $c_{new}=(\tfrac{p}{t})^{t/p} \ c_t^{t/p} \ (4c_p)^{1-t/p}$.
\end{claim}
\begin{proof}
With probability at least $\frac{5}{6}$, $A_{base}$ outputs a point $x$ with $\|x-q\|_p \leq r c_p$.
By triangle inequality, $\|x^*-x\|_p\leq \|x^*-q\|_p+\|q-x\|_p\leq 2rc_p$, hence $\|M^x(x^*)-M^x(q)\|_t\leq r$.
Thus, with probability at least $\frac{5}{6}$, $A_x$ outputs a point $M^{x}(z)$ with $\|M^x(z)-M^x(q)\|_t\leq r c_t$. By a union bound, both events hold with probability $2/3$. Assume they hold.
By \Cref{thm:mazur_map}, 
\[
\frac{t}{p} \cdot {(4 rc_p)^{1-p/t}}\|z-q\|_{p}^{p/t} \leq \|M^{x}(z)-M^{x}(q)\|_t \leq r \cdot c_t,
\] 
rearranging this we obtain $\|z-q\|_p \leq r(\tfrac{p}{t})^{t/p} \ c_t^{t/p} \ (4c_p)^{1-t/p} \equiv r \cdot c_{new}$.
\end{proof}

\begin{remark}
    Plugging $t=2$ into \Cref{claim:ANN_Mazur_Map} and solving the recursion, we obtain a variation of \cite[Lemma 11]{BG19}. 
\end{remark}

Now, as in the proof of \Cref{thm:Optimal-Lipschitz-Decompositions-in-lp}, we apply the additional recursive embedding reduction that goes through intermediate $\ell_t$ spaces. 
    To improve readability, we first provide a simpler proof with $O(p^3)$-approximation, and then explain the improvement to $O(p^{1+\ln(4)+\eps})$-approximation.
    We assume without loss of generality that $p\leq\log{d}$ by H\"older's inequality.

    Assume for now that $p$ is a power of $2$. 
    Consider the data structure for $\ell_2^d$ given by \cite{Chan98}, with approximation $c=\poly(d)$, space and processing time $\tilde{O}(n\cdot \poly(d))$ and query time $\poly(d\log n)$. 
    By H\"older's inequality, the same data structure yields $\poly(d)$ approximation also for $\ell_p^d$.

    Now, we recursively apply \Cref{claim:ANN_Mazur_Map} with $t=p/2$, as follows.
    Denote by $k$ the number of recursive steps to be determined later, and by $\hat{c}_i$ the approximation guarantee in $\ell_p$ after the $i$-th recursive step.
    Initially, $\hat{c}_0 = \poly(d)$, by using the data structure of \cite{Chan98}.
    For every $i\in [k]$, we maintain data structures $\{A_x^i\}_{x\in V}$, where the Mazur map is scaled according to the current approximation guarantee (i.e., scaled down by $\frac{p}{t} \cdot (2r \hat{c}_{i-1})^{p/t-1}$).
    Moreover, we amplify the success probabilities to $1-\tfrac{2}{3k}$ by $O(\log k)$ independent repetitions.
    Thus, if the $(i-1)$-th iteration is successful, i.e., it returns a point $x$ solving $(\hat{c}_{i-1},r)$-ANN, then the Mazur maps in the $i$-th iteration are scaled correctly. Hence, by querying $A_x^i$, we get the approximation given by \Cref{claim:ANN_Mazur_Map}.
    By the law of total probability, with probability $2/3$, all the $k$ recursive steps return a correct estimate. Therefore,

    \begin{align}
        \hat{c}_k(V) &\leq  \sqrt{8c_{p/2}\cdot \hat{c}_{k-1}(V)} \nonumber \\
        &\leq  \sqrt{8c_{p/2}\cdot  \sqrt{8c_{p/2}\cdot \hat{c}_{k-2}(V)}} \nonumber \\
        &\leq \cdots \nonumber \\
        &\leq 
        (8c_{p/2})^{(1/2+1/4+\ldots+1/2^k)} \ \hat{c}_{0}(V)^{2^{-k}} \nonumber \\
        &\leq 8 c_{p/2}\cdot \hat{c}_0(V)^{2^{-k}}. 
    \end{align}
Picking
$k := \ceil{\log(\log p\log \hat{c}_0(V))} = O( \log\log \log {d})$ yields $\hat{c}_0(V)^{2^{-k}} \leq 2^{1/\log p}$, and we obtain a data structure with approximation at most $\hat{c}_k(V) \leq 2^{3+1/\log p} \cdot c_{p/2}$. 

Before applying a second recursion on $p$, we amplify the success probabilities to $1-\tfrac{2}{3\log p}$ by $O(\log \log p)=O(\log\log\log d)$ independent repetitions.
Now a second recursion on $p$ implies $\hat{c}_k(V)\leq 2 p^3 \cdot c_2$ with probability at least $2/3$. 
Finally, we bound $c_2$ similarly to~\cite{BG19}, 
namely, using the JL-lemma to reduce the dimension to $O(\log{n})$ 
together with a $(2,r)$-ANN data structure of~\cite{KOR00, HIM12} in $\ell_2^{O(\log{n})}$, 
which has query time $T_2 = \polylog{n}$, and space and preprocessing time $S_2=Z_2=n^{O(1)}$.
Plugging this as the base case of the second recursion,
and we get the desired approximation $\hat{c}_k(V)=O(p^3)$. 
Each level of the second recursion increases the space and preprocessing time by factor $n$, 
resulting in a total of $n^{O(\log p)}\cdot S_2 = n^{\log p+O(1)}\cdot d^{O(1)}$ space and preprocessing time. %
Answering a query goes through both recursions, but the first recursion only requires $O(k\log k)=\tilde{O}(\log\log\log{d})$ calls to an ANN data structure for $\ell_t$,
hence the overall running time is
$(\log\log\log{d})^{O(\log{p})} \cdot T_2=\poly(d\log n)$. 
Resolving the case when $p$ is not a power of $2$ is straightforward and performed exactly as in the proof of \Cref{thm:Optimal-Lipschitz-Decompositions-in-lp}, and thus omitted.

\medskip

To improve the approximation, let $\eps>0$, and pick $t=(1-\eps)p$ instead of $t=p/2$.
We now have that 
\[\hat{c}_k(V) \leq (\tfrac{1}{1-\eps})^{1-\eps} \ c_{t}^{1-\eps} \ (4\hat{c}_{k-1}(V))^{\eps} \leq \ldots \leq (\tfrac{c_{t}}{1-\eps})^{1-\eps^{k}} 4^{\tfrac{\eps (1-\eps^{k})}{1-\eps}}(\hat{c}_0(V))^{\eps^{k}}.\]
For sufficiently large $k=O(\log(\eps^{-1}) \log(\log p \log d))$, we get $\hat{c}_k(V)\leq \tfrac{1}{1-\eps} 4^{\tfrac{\eps}{1-\eps}} c_t$.
Now, a recursion on $p$ for $\log_{\tfrac{1}{1-\eps}}p=O(\eps^{-1}\log p)$ levels implies
\[
\hat{c}_k(V)\leq p\cdot \exp\big(\ln(4)(\tfrac{\eps}{1-\eps}\cdot \log_{\tfrac{1}{1-\eps}}p)\big) c_2
\leq p^{1+\ln(4)+O(\eps)} c_2,
\]
where the last step uses the inequalities $\tfrac{1}{1-\eps}\geq 1+\eps$ and $\ln(1+\eps)\geq \tfrac{\eps}{1+\eps}$.
The rest of the proof is the same, and the space and preprocessing time increase to $\poly(d n^{\eps^{-1}\log p})$.
Rescaling $\eps$ concludes the proof.
\end{proof}

%% file: l_2-distortion.tex
\section{Embedding Finite $\ell_p$ Metrics into $\ell_2$}
\label{sec:Embedding-finite-l_p-metrics-into-l_2}
In this section, we prove \Cref{thm:improved-c_2-distortion}
by providing embeddings of finite $\ell_p$ metrics into $\ell_2$, for $3<p<3\sqrt{e}$. 
We will need the following setup from~\cite{NR25}.

\begin{definition}[Definition 4 in \cite{NR25}]
\label{def:localized-weakly-bi-Lipschitz-embedding} 
Given \( K, D > 1 \), we say that a metric space \( (\M, d_\M) \) admits a 
\emph{\( K \)-localized weakly bi-Lipschitz embedding} into a metric space \( (\mathcal{N}, d_\mathcal{N}) \) with distortion \( D \) if for every \( \Delta > 0 \) and every subset \( \mathcal{C} \subseteq \M \) of diameter \( \operatorname{diam}_\M(\mathcal{C}) \leq K \Delta \), there exists a non-constant Lipschitz function  
\( f_\Delta^\mathcal{C}: \mathcal{C} \to \mathcal{N} \) satisfying the following.
For every $x,y\in\mathcal{C}$, if $d_\M(x, y) > \Delta$, then
\[
  d_\mathcal{N} \left( f_\Delta^\mathcal{C}(x), f_\Delta^\mathcal{C}(y) \right) > \frac{\| f_\Delta^\mathcal{C} \|_{\text{Lip}}}{D} \Delta,
\]
where \( \| \cdot \|_{\text{Lip}} \) is the Lipschitz constant.
\end{definition}

We provide the following simple observation, that composing a localized weakly bi-Lipschitz embedding with a low-distortion embedding yields a localized weakly bi-Lipschitz embedding, as follows.
\begin{observation}\label{obs:localized-bi-Lipschitz-composed-with-bi-Lipschitz}
Let \( (\M, d_\M),(\mathcal{N}, d_\mathcal{N}),(\mathcal{Z}, d_\mathcal{Z}) \) be metric spaces, such that $(\M, d_\M)$ admits a $K$-localized weakly bi-Lipschitz embedding into \( (\mathcal{N}, d_\mathcal{N}) \) with distortion \( D_1 \) and $(\mathcal{N}, d_\mathcal{N})$ admits an embedding into $(\mathcal{Z}, d_\mathcal{Z})$ with distortion $D_2$. Then $(\M, d_\M)$ admits a $K$-localized weakly bi-Lipschitz embedding into \( (\mathcal{Z}, d_\mathcal{Z}) \) with distortion \( D_1 \cdot D_2 \).       
\end{observation}
\begin{proof}
Let $\Delta>0$ and \( \mathcal{C} \subseteq \M \) of diameter \( \operatorname{diam}_\M(\mathcal{C}) \leq K \Delta \). Let \( f_\Delta^\mathcal{C}: \mathcal{C} \to \mathcal{N} \) be the function promised by \Cref{def:localized-weakly-bi-Lipschitz-embedding}, and $g:(\mathcal{N}, d_\mathcal{N}) \to (\mathcal{Z}, d_\mathcal{Z})$ be an embedding with distortion $D_2$. Consider $\Tilde{f}_\Delta^\mathcal{C} := g \circ f_\Delta^\mathcal{C}$.
Recall that since $g$ has distortion at most $D_2$, there exists $s>0$ such that for every $u,v\in \N$, we have
$\frac{s}{D_2} \cdot d_\N(u,v) \leq d_\mathcal{Z}(g(u), g(v)) \leq s \cdot d_\N(u, v)$.
Since $f_\Delta^\mathcal{C}$ is non-constant and $g$ has bounded contraction, $\Tilde{f}_\Delta^\mathcal{C}$ is non-constant. Let $x,y\in\mathcal{C}$ such that $d_\M(x, y) > \Delta$.  Hence,
\[
d_\mathcal{Z} \left( \Tilde{f}_\Delta^\mathcal{C}(x), \Tilde{f}_\Delta^\mathcal{C}(y) \right) \geq \frac{s}{D_2} \cdot d_\mathcal{N} \left( f_\Delta^\mathcal{C}(x), f_\Delta^\mathcal{C}(y) \right) > \frac{s \cdot \| f_\Delta^\mathcal{C} \|_{\text{Lip}}}{D_1 \cdot D_2} \Delta,
\]
where the last inequality follows since $f_\Delta^\mathcal{C}$ is a $K$-localized weakly bi-Lipschitz embedding with distortion $D_1$. Since $g$ expands distances by at most a factor $s$, we have $\| \Tilde{f}_\Delta^\mathcal{C} \|_{\text{Lip}} \leq s \cdot\| f_\Delta^\mathcal{C} \|_{\text{Lip}}$, concluding the proof.
\end{proof}

\begin{lemma}[Generalization of Lemma 5 in \cite{NR25}]
\label{lem:localized-properties-Mazur-map}
For every \( K > 1 \), if \( p > q \geq  1 \), then \( \ell_p \) admits a \( K \)-localized weakly bi-Lipschitz embedding into \( \ell_q \)  
with distortion
$
O_{p/q} (K^{p/q-1}).
$
\end{lemma}
\begin{proof}
Fixing $K,\Delta$>0 and a subset $\mathcal{C} \subset \ell_p$ whose $\ell_p$ diameter is at most $K\Delta$, 
pick an arbitrary point $z \in \mathcal{C}$, and consider the Mazur map $M_{p,q}$ scaled down by $(K\Delta)^{p/q-1}$ on $\C-z$.
The lemma follows immediately by \Cref{thm:mazur_map}.
\end{proof}

\begin{definition}
\label{def:extension-modulus}
The Lipschitz extension modulus \( e(\M,\mathcal{N}) \) of a pair of metric spaces \( \M, \mathcal{N} \) is the infimum over all \( L \in [1, \infty) \) such that for every subset \( \mathcal{C} \subseteq \M \), every \( 1 \)-Lipschitz function \( f: \mathcal{C} \to \mathcal{N} \) can be extended to an \( L \)-Lipschitz function \( F: \M \to \mathcal{N} \).
\end{definition}

\begin{theorem}[Theorem 6 in \cite{NR25}] \label{thm:localized-bi-Lipschitz-imply-bounded-c_2-distortion}
There is a universal constant \( \kappa > 1 \) with the following property. Fix \( \theta > 0 \), an integer \( n \geq 3 \), and \( \alpha > 1 \). 
Let \( (\M, d_\M) \) be an \( n \)-point metric space such that every subset \( \mathcal{C} \subseteq \M \) with \( |C| \geq 3 \) admits a  
\( \kappa (\log |\mathcal{C}|) \)-localized weakly bi-Lipschitz embedding into \( \ell_2 \) with distortion \( \alpha (\log |\mathcal{C}|)^\theta \). Then  

\[
c_2(\M) \leq \alpha \cdot e(\M; \ell_2)\cdot (\log n)^{\max \left\{ \theta, \frac{1}{2} \right\} }\cdot \log \log n.
\]    
\end{theorem}

Next, we show a reduction that takes embeddings of finite $\ell_q$ metrics into $\ell_2$, and constructs an embedding of finite $\ell_p$ metric into $\ell_2$, for $p>q$.
The proof constructs a localized weakly bi-Lipschitz embedding of $\ell_p$ into $\ell_q$ and composes it with the given embedding from $\ell_q$ into $\ell_2$.
By \Cref{obs:localized-bi-Lipschitz-composed-with-bi-Lipschitz}, this yields a localized weakly bi-Lipschitz embedding from $\ell_p$ into $\ell_2$, and by \Cref{thm:localized-bi-Lipschitz-imply-bounded-c_2-distortion}, we get a low-distortion embedding into $\ell_2$.

For every $q \in [1,\infty]$, define 
\[
    \xi_q \coloneqq \inf_{\theta \geq 0}  \Big\{\theta : \exists \nu>0, \forall n \geq 2, \quad c_2^n(\ell_q) \leq \nu \cdot \log^\theta{n} \Big\},
\]
where $\xi_q \leq 1$ for all $q \in [1,\infty]$ by Bourgain's embedding \cite{Bourgain85}.

\begin{lemma}
\label{lem:recursive-l_2-embedding}
For every $2 \leq q<p$,
\[
    \xi_{p} \leq \max\{\tfrac{1}{2}, \xi_q\}+\tfrac{p}{q}-1.
\]
\end{lemma}
\begin{proof}
Let $\delta>0$ and let $\mathcal{\M} \subset \ell_p$ be an $n$-point metric.
If $n \leq 2$, then clearly $c_2^n(\ell_p)=1$. Otherwise, let $\mathcal{C} \subseteq \M$ with $|\mathcal{C}| \geq 3$.
We now construct a weakly bi-Lipschitz embedding of $\C$ into $\ell_2$. By \Cref{lem:localized-properties-Mazur-map,obs:localized-bi-Lipschitz-composed-with-bi-Lipschitz}, we have that for every $K\geq 1$, $\C$ admits a $K$-localized weakly bi-Lipschitz embedding into $\ell_2$ with distortion $O(K^{p/q-1} \cdot c_2^{|\mathcal{C}|}(\ell_q))$. 
Setting $K=\kappa(\log{|\mathcal{C}|})$, where $\kappa$ is the universal constant from \cref{thm:localized-bi-Lipschitz-imply-bounded-c_2-distortion}, and using $c_2^{|\mathcal{C}|}(\ell_q) \leq O_\delta( \log^{\xi_q+\delta}{|\mathcal{C}|})$, we obtain a $\kappa(\log|\mathcal{C}|)$-localized weakly bi-Lipschitz embedding of $\mathcal{C}$ into $\ell_2$ with distortion $O_{p,\delta}(\log^{\frac{p}{q}-1+\xi_q+\delta}{|\mathcal{C}|})$.

By \Cref{thm:localized-bi-Lipschitz-imply-bounded-c_2-distortion},
\begin{align*}
    c_2(\ell_p) &\leq O_{p,\delta}\Big( e(\ell_p; \ell_2) (\log n)^{\max\{\frac{1}{2}, \frac{p}{q}-1+\xi_q+\delta\}} \log \log n\Big) \\
    &\leq O_{p,\delta}\Big( (\log n)^{\max\{\frac{1}{2}, \frac{p}{q}-1+\xi_q+\delta\}} \log \log n\Big) && \text{$e(\ell_p,\ell_2) \leq O(\sqrt{p})$ by \cite{NPSS06}} \\
    &\leq O_{p,\delta}\Big( (\log n)^{\max\{\frac{1}{2}, \xi_q\}+\frac{p}{q}-1+\delta}\log\log n \Big) && \text{since $\frac{p}{q}-1+\delta > 0$} \\
    &\leq O_{p,\delta}\Big((\log n)^{\max\{\frac{1}{2}, \xi_q\}+\frac{p}{q}-1+2\delta}\Big).
\end{align*}
Since $\delta$ is arbitrary, the lemma follows.

\end{proof}

The reduction given in the lemma above is a single iteration of recursive embedding, and we repeat it recursively to prove \Cref{thm:improved-c_2-distortion}.

\begin{proof}[Proof of \Cref{thm:improved-c_2-distortion}]
    Let $3<p<3\sqrt{e}$ and $\eps>0$. 
    Consider a sequence $q_0,\ldots,q_k$, where $q_0=p$ and $\tfrac{q_i}{q_{i+1}}=(\tfrac{p}{3})^{1/k}$ for all $i\in [0,k-1]$. Therefore, $q_k=3$.
    By \cref{lem:recursive-l_2-embedding} we have,
    \begin{align*}
        \xi_p & \leq \max\{\tfrac{1}{2}, \xi_{q_1}\}+\tfrac{p}{q_1}-1 \\
        & \leq \max\{\tfrac{1}{2}, \xi_{q_2}\}+\tfrac{p}{q_1}-1+\tfrac{q_1}{q_2}-1\\
        &\ldots \\
        &\leq \max\{\tfrac{1}{2}, \xi_{3}\}+(\tfrac{p}{q_1}-1 + \tfrac{q_1}{q_2}-1+\ldots +\tfrac{q_{k-1}}{q_k}-1).
    \intertext{By~\cite[Theorem 1]{NR25}, we have $c_2^n(\ell_3) \leq O(\sqrt{\log{n}}\cdot\log\log{n})$, and thus $\xi_3 \leq \frac{1}{2}$. Therefore,}
    &= \tfrac{1}{2}-k+\sum_{i=0}^{k-1} \tfrac{q_i}{q_{i+1}} \\
    &=\tfrac{1}{2}-k + k (\tfrac{p}{3})^{1/k}
    =\tfrac{1}{2}-k + k\cdot \exp(\tfrac{1}{k}\ln \tfrac{p}{3}).
    \intertext{For a suitable choice of $k=O(\eps^{-1})$, and using the useful inequality $e^x\leq 1+x+x^2$ for $x<1.79$,}
    &\leq \tfrac{1}{2}-k + k \big(1+ \tfrac{1}{k}\ln \tfrac{p}{3} + (\tfrac{1}{k}\ln \tfrac{p}{3})^2\big) \\
    &< \tfrac{1}{2} + \ln \tfrac{p}{3} + \eps.
    \end{align*}
    The theorem follows from the definition of $\xi_p$.
\end{proof}

%% file: Future_Directions.tex
\section{Future Directions}
\label{sec:conclusion}

\subparagraph*{Problems in $\ell_p$, $p < 2$. }  
Our results for \( \ell_p \) spaces are all for \( p > 2 \).
For the other case, \( p < 2 \),
there are natural candidates for intermediate spaces,
namely, \( \ell_q \) for \( p < q < 2 \).
Can recursive embedding be used in such settings?

\subparagraph*{Problems in $\ell_\infty$.  }
Many problems in $\ell_\infty^d$ can be reduced to $\ell_2^d$
using John's theorem~\cite{John48},
which incurs $O(\sqrt{d})$ multiplicative distortion and is known to be tight.
Our method bypasses this limitation
and reduces the Lipschitz decomposition problem from $\ell_\infty^d$ to $\ell_2^d$
at the cost of only a polylogarithmic (in $d$) factor.
Indeed, the reduction in \Cref{thm:Optimal-Lipschitz-Decompositions-in-lp}
actually proves (although not stated explicitly) that
\begin{equation}
  \label{eq:reduce_infty_to_2}
  \beta^*(\ell_\infty^d) \leq \polylog(d) \cdot \beta^*(\ell_2^d).
\end{equation}
Can other problems in \( \ell_\infty^d \) be resolved similarly,
i.e., through a recursive embedding to \( \ell_2^d \)
that bypasses the $O(\sqrt{d})$ factor of a direct embedding?

\subparagraph*{Lower Bounds.}  
Our approach of reducing from $\ell_\infty^d$ to $\ell_2^d$
can also establish lower bounds for problems in $\ell_2^d$,
which essentially amounts to ``pulling'' hard instances,
from \( \ell_\infty^d \) into \( \ell_2^d \).
For $\beta^*(\ell_2^d)$, a tight bound is already known~\cite{CCGGP98},
and thus~\eqref{eq:reduce_infty_to_2} cannot yield a new lower bound for it.
However, for the extension modulus of $\ell_2^d$,
the known bounds are not tight,
namely, $\Omega(d^{1/4}) \leq e(\ell_2^d) \leq O(\sqrt{d})$~\cite{LN05, MN13},
and it is conjectured that $e(\ell_2^d)=\Theta(\sqrt{d})$~\cite{Naor17}.
Can the known lower bound $e(\ell_\infty^d) \geq \Omega(\sqrt{d})$ 
be pulled to $\ell_2^d$, analogously to~\eqref{eq:reduce_infty_to_2}?

\subparagraph*{Nearest Neighbor Search.}  
The space and preprocessing time of our data structure in \cref{thm:BNN-for-l_p}
are not polynomial in $n$ and $d$ whenever \( p \) is non-constant. 
This increase in preprocessing time and space was somewhat mitigated
in~\cite{BG19} in the special case of doubling metrics.
Can this issue be avoided also in the general case?

\subparagraph*{Low-Distortion Embeddings.}  
There remains a gap in our understanding of the distortion required to embed
finite $\ell_p$ metrics into $\ell_2$ for every \( p \in (2, \infty) \).
For the special case of doubling metrics,
we know from~\cite[Theorem 5.5]{BG14_v2} that
$
  c_2(\mathcal{C})
  \leq
  O\left(\sqrt{\ddim(\mathcal{C})^{p/2 - 1} \log {n}}\right)
$
for every \( p \in (2,\infty) \)
and every \( n \)-point metric \( \mathcal{C} \subset \ell_p \),
where $\ddim(\mathcal{C})$ denotes its doubling dimension.
This upper bound above does not match
the $\Omega(\log^{1/2-1/p}{n})$ lower bound in \Cref{rem:embedding-lower-bound},
which actually holds for doubling metrics.  
We thus ask whether the distortion bound in the doubling case be improved.